\documentclass[runningheads,envcountsame]{llncs}
\usepackage{graphicx}
\usepackage[english]{babel}
\usepackage{color}

\usepackage{amsthm}
\usepackage{amsmath}
\usepackage{amsfonts}
\usepackage{latexsym}
\usepackage{amssymb}
\usepackage{enumerate}
\usepackage{multirow}
\usepackage{mathtools}

\usepackage{todonotes}
\usepackage{comment}
\usepackage{refcount}
\newcounter{myCounter}
\newcommand{\oR}{\overline{R}}
\newcommand{\oq}{\overline{q}}
\theoremstyle{definition}
\newtheorem*{proofofProp}{Proof of Proposition \themyCounter}
\newtheorem*{proofofLemma}{Proof of Lemma \themyCounter}
\newtheorem*{proofofTh}{Proof of Theorem \themyCounter}
\begin{document}
\title{Logarithmic equal-letter runs for $BWT$\\ of purely morphic words}
%
%
\author{A. Frosini\inst{1} \and
I. Mancini\inst{2} \and
S. Rinaldi\inst{2}
\and
G. Romana\inst{3}
\and
M. Sciortino\inst{3}
}
\authorrunning{A. Frosini et al.}
%
\institute{Università di Firenze, Italy \\
\email{andrea.frosini@unifi.it}\and
Università di Siena, Italy \\
\email{ilaria.mancini@student.siena.it}, \email{simone.rinaldi@unisi.it}
\and
Università di Palermo, Italy\\
\email{\{giuseppe.romana01,marinella.sciortino\}@unipa.it}}
\maketitle              
\def\rle(#1){\texttt{rle}(#1)}
\newcommand{\bwt}{ {\tt bwt}}

\begin{abstract}
In this paper we study the number $r_{\bwt}$ of equal-letter runs produced by the \emph{Burrows-Wheeler transform} ($BWT$) when it is applied to \emph{purely morphic finite words}, which are words generated by iterating prolongable morphisms. Such a parameter $r_{\bwt}$ is very significant since it provides a measure of the performances of the $BWT$, in terms of both compressibility and indexing. In particular, we prove that, when $BWT$ is applied to any purely morphic finite word on a binary alphabet, $r_{\bwt}$ is $\mathcal{O}(\log n)$, where $n$ is the length of the word. Moreover, we prove that $r_{\bwt}$ is $\Theta(\log n)$ for the binary words generated by a large class of prolongable binary morphisms. These bounds are proved by providing some new structural properties of the \emph{bispecial circular factors} of such words.
\keywords{Burrows-Wheeler Transform  \and Equal-letter runs \and Morphisms \and Bispecial circular factors}
\end{abstract}
\section{Introduction}

The Burrows-Wheeler Transform ($BWT$) is a reversible transformation that produces a permutation of the text given in input, according to the lexicographical order of its cyclic rotations. It was introduced in 1994 in the field of Data Compression \cite{BW94} and it still represents the main component of some of the most known lossless text compression tools \cite{SewardBzip2,cj/Fenwick96}
as well as of compressed indexes \cite{Ferragina:2005}. 
$BWT$ is used as pre-processing of memoryless compressors, causing the \emph{boosting} of their performance. 
The key motivation for this fact is that $BWT$ is likely to create equal-letter runs (clusters) that are longer than the clusters of the original text. 
In other words, if we denote with $r(w)$ the number of equal-letter runs in the word $w$, the number $r_{\bwt}(w)$ of equal-letter runs produced by $BWT$ applied to $w$ often becomes lower than $r(w)$. 
It is important to note that the performance in terms of both space and time of text compressors and compressed indexing data structures applied on a text $w$ can be evaluated by using $r_{\bwt}(w)$ \cite{GagieNP20}. 
An upper bound on the number of clusters produced by $BWT$ has been provided in \cite{KempaK20}, in particular, it has been proved that $r_{\bwt(w)}=\mathcal{O}(z(w)\log^2 n)$ where $z(w)$ is the number of phrases in the $LZ77$ factorization of $w$ and $n$ is the length of $w$. 
The ratio between $r_{\bwt}(w)$ and the number of clusters in the $BWT$ of the reverse of $w$ has been studied in \cite{GiulianiILPST21}. 
A recent comparative survey illustrating the properties of $r_{\bwt}(w)$ and other repetitiveness measures can be found in \cite{Navarro21a}. In particular, in this survey the measure $\gamma$, which is the size of the smallest string attractor for the sequence \cite{KempaP18}, and the measure $\delta$, which is defined from the factor complexity function \cite{ChristiansenEKN21}, are also considered. 
From a combinatorial point of view, the parameter $r_{\bwt}$ has been studied in order to obtain more information about the combinatorial complexity of a word from the number of clusters produced by applying the $BWT$. 
In particular, great attention has been given to the characterization of the words for which  the $BWT$ produces the minimal number of clusters \cite{MaReSc,puglisiSimpson2008,RestivoRosoneTCS2009,Ferenczi_Zamboni2013}. A first combinatorial investigation of the $BWT$ clustering effect has been given in \cite{MantaciRRS_words2017,MANTACI_TCS2017} in which the \emph{$BWT$-clustering ratio} $\rho(w)=\frac{r_{\bwt}(w)}{r(w)}$ has been studied. 
In particular, it has been proved in \cite{MANTACI_TCS2017} that $\rho(w)\leq 2$ and infinite families of words for which $\rho$ assumes its maximum value have been shown. In \cite{MantaciRRS_words2017} the behavior of $\rho$ is studied for two very well known families of words, namely Sturmian words and de Brujin binary words.

This paper is focused on investigating the behaviour of $BWT$ when applied to finite words obtained by iterating a morphism. Morphisms are well-known objects in the field of Combinatorics on Words and they represent a powerful and natural tool to define repetitive sequences. Studying the compressibility of repetitive sequences is an issue that is raising great interest. The morphisms, combined with macro-schemes, have been used to define other mechanisms to generate repetitive sequences, called $NU$-systems \cite{NavarroU21}.  We consider the morphisms $\varphi$ that admit a fixed point (denoted by $\varphi^{\infty}(a)$) starting from a given character $a\in A$, i.e. $\varphi^{\infty}(a)=\lim_{i\to\infty}\varphi^i(a)$. Such morphisms are called \emph{prolongable} on $a$. 
In \cite{shallit_shaeffer2020} the measure $\gamma$ is computed for the prefixes of infinite words that are fixed points of some morphisms. Moreover, a complete characterization of the Lempel-Ziv complexity $z$ for the prefixes of fixed points of prolongable morphisms has been given in \cite{ConstantinescuI07}.
In this paper we analyse the number $r_{\bwt}$ of equal-letter runs and the $BWT$-clustering ratio $\rho$ when $BWT$ is applied to the \emph{purely morphic finite words}, i.e. the words $\varphi^i(a)$ generated by iterating a morphism $\varphi$ prolongable on $a$. 
In \cite{BrlekFMPR19}, the parameter $\rho$ has been computed for the families of finite words generated by some morphisms.
In all cases considered in \cite{BrlekFMPR19}, the $BWT$ efficiently clusters since the value of $\rho$ is much lower than $1$. In the one-page abstract appeared in \cite{NoiDCC} we extended such results by providing some new upper bounds on $r_{\bwt}(\varphi^i(a))$ depending on the factor complexity of the fixed point $\varphi^\infty(a)$. 

In this paper, we define the notion of \emph{$BWT$-highly compressible} morphism by evaluating, for a given morphism $\varphi$, whether $BWT$-clustering ratio $\rho$ on the words $\varphi^i(a)$ tends towards zero, when $i$ goes to infinity, and we identify some classes of $BWT$-highly compressible morphisms, by proving that $r_{\bwt}$ is $\mathcal{O}(\log n)$ for words of length $n$ generated by any primitive morphism.
Moreover, 
we give some combinatorial properties of several classes of binary morphisms and we improve the results announced in \cite{NoiDCC} in the case of binary purely morphic words. In particular, we prove that for any word $w$ on the binary alphabet $\{a,b\}$ generated by iterating a prolongable morphism, 
the parameter $r_{\bwt}(w)$ is $\mathcal{O}(\log n)$, where $n$ is the length of $w$. 
A consequence of these results is that such morphisms, except a few cases, are $BWT$-highly compressible. 
Finally, we prove that $r_{\bwt}(w)$ is $\Theta(\log n)$ for the binary finite words $w$ generated by a large class of prolongable morphisms. 
Such bounds are obtained by using a close relation between $r_{\bwt}(w)$ and the combinatorial notion of a bispecial circular factor of $w$ and by providing some new structural properties of the bispecial circular factors of infinite families of finite binary words generated by prolongable morphisms. 


\section{Preliminaries}
Let $A =\{a_1, a_2, \ldots, a_k\}$ be a finite ordered alphabet with $a_1< a_2< \ldots < a_k$, where $<$ denotes the standard lexicographic order.
We assume that $|A|\geq 2$. The set of words over the alphabet $A$ is denoted by $A^*$.
A finite word $w=w_1w_2\cdots w_n\in A^*$ is a finite sequence of letters from $A$.
The length of $w$, denoted $|w|$, is the number $n$ of its letters, $|w|_a$ denotes the number of occurrences of the letter $a$ in $w$.
An infinite word $x=x_1 x_2 x_3\ldots$ is a non-ending sequence of elements of the alphabet $A$.  

Given an infinite or finite word $x$, we say that a word $u$ is a \emph{factor} of $x$ if $x=vuy$ for some words $v$ and $y$.
The word $u$ is a \emph{prefix} (resp. \emph{suffix}) of $x$ if $x=uy$ (resp. $x=yu$) for some word $y$.
A factor $u$ of $x$ is \emph{left special} (\emph{right special}) if there exist $a,b \in A$ with $a \neq b$ such that both $au$ and $bu$ ($ua$ and $ub$) are factors of $x$.
A factor $u$ is \textit{bispecial} if it is both left and right special.
We denote by $f_x(k)$ the number of distinct factors of $x$ having length $k$. The function $f_x$ is called \emph{factor complexity} of $x$.

We say that a finite word $w$ has a \emph{period} $p>0$ if $w_i=w_{i+p}$ for each $i\leq |w|-p$. 
It is easy to see that each integer $p\geq |w|$ is a period of $w$.
The smallest of such integers is called \emph{minimum period} of $w$. The notion of period can be also given for infinite words.
We say that an infinite word is \emph{ultimately periodic} with period $p>0$ if exists $K \geq 1$ such that $w_i=w_{i+p}$ for each $i \geq K$.
Moreover, if this condition holds for any $i \in \mathbb{N}$, $w$ is said \emph{periodic} (with period $p$).
An infinite word $x$ is \emph{aperiodic} if it is not ultimately periodic. 

Given two finite words $w,z\in A^*$, we say that $w$ is a cyclic rotation of $z$, or equivalently $w$ and $z$ are {\em conjugate}, if $w=uv$ and $z=vu$, where $u,v\in A^*$.
Conjugacy between words is an equivalence relation over $A^*$.
We say that a finite word $u$ is a \emph{circular factor} of $w$ if $u$ is a factor of a conjugate of $w$.
For instance, $aa$ is a circular factor of $abbba$, but it is not a factor. We denote by $\mathcal{C}(w)$ the set of circular factors of a word $w$.
If we denote by $c_w(k)$ the number of distinct circular factors of $w$ having length $k$, it is easy to see that $f_w(k)\leq c_w(k)$, for each $k\geq 1$.
Note that the notions of left special, right special and bispecial factors can be given circularly, considering circular factors instead of factors of a given word. In particular, we say that $u$ is a \emph{bispecial circular factor} of a word $w$ if there exist $a,b\in A$, with $a\neq b$, and $a',b'\in A$, with $a'\neq b'$, such that both $aua'$ and $bub'$ are circular factors of $w$. We denote by $BS(w)$ the set of bispecial circular factors of $w$. 

The {\em Burrows-Wheeler Transform} ($BWT$) is a reversible transformation introduced in the context of Data Compression \cite{BW94}.
Given a word $w\in A^*$, the $BWT$ produces a permutation of $w$ which is obtained by concatenating the last letter of the lexicographically sorted cyclic rotations of $w$, and we denote it with $\bwt(w)$. Note that $\bwt(w)=\bwt(v)$ if and only if $w$ and $v$ are conjugate.

The {\em run-length encoding} of a word $w$, denoted by $\rle(w)$, is a sequence of pairs
$(w_i, l_i$) with $w_i \in A$ and $l_i > 0$, such that $w = w_1^{l_1}w_2^{l_2}\ldots w_r^{l_r}$ and $w_i \neq w_{i+1}$.
We denote by $r(w)=|\rle(w)|$, i.e. the number $r$ of equal-letter runs in $w$. We denote by $r_{\bwt}(w)=r(\bwt(w))$ the number of equal-letter runs in $\bwt(w)$.

Morphisms are fundamental tools of formal languages and a very crucial notion in combinatorics on words. They represent a very interesting way to generate an infinite family of words. Let $A$ and $\Sigma$ be alphabets. A \emph{morphism} is a map $\varphi$ from $A^*$ to $\Sigma^*$ that obeys the identity $\varphi(uv) = \varphi(u)\varphi(v)$ for all words $u, v \in A^*$. 
By definition, a morphism can be described by just specifying the images of the letters of $A$.
Examples of very well known morphisms are the \emph{Thue-Morse morphism} $\tau$, defined as $\tau(a)=ab$ and $\tau(b)=ba$, and the \emph{Fibonacci morphism} $\theta$, defined as $\theta(a)=ab$ and $\theta(b)=a$.
A morphism $\varphi$ is \emph{primitive} if there exists a positive integer $k$ such that, for every pair of characters $a,b \in A$, the character $a$ occurs in $\varphi^k(b)$. 
Both $\tau$ and $\theta$ are primitive morphisms. 
A morphism is called \emph{non-erasing} if $|\varphi(a)|\geq 1$, for each $a\in A$. We assume to consider non-erasing morphisms. 

Morphisms can be classified by the length of images of letters.
If there is a constant $k$ such that $|\varphi(a)| = k$ for all $a \in A$ then we say that $\varphi$ is $k$-\textit{uniform} (or just uniform, if $k$ is clear from the context).
For instance, the Thue-Morse morphism $\tau$ is $2$-uniform.  The \emph{growth function} of a morphism $\varphi$ with respect to a letter $a \in A$ and an iteration $i$ is defined by $\varphi_a(i) = |\varphi^i(a)|$.
A letter $a$ is said to be \emph{growing} for $\varphi$ if $\lim_{i\to\infty}\varphi_a(i)=+\infty$, otherwise it is $\emph{bounded}$. A morphism $\varphi$ is growing if each letter of the alphabet is growing for $\varphi$. For growing morphisms it holds that, for any $a \in A$, $\varphi_a(i) = \Theta(i^{e_a}p_a^i)$, for some $e_a\geq 0$ and $p_a>1$. Another classification of morphisms is according to its growth function on distinct letters \cite{Pansiot_ICALP84}.
A growing morphism $\varphi$ is called \emph{quasi-uniform} if $\varphi_a(i) = \Theta(p^i)$ for any $a \in A$ and some $p > 0$; $\varphi$ is called \emph{polynomially divergent} if for any $a \in A$ it holds that $\varphi_a(i) = \Theta(i^{e_a}p^i)$ for some $p > 1$ and exist $a,b \in A$ such that $e_a \neq e_b \geq 0$;
$\varphi$ is called \emph{exponentially divergent} if exist $a,b \in A$ such that $\varphi_a(i) = \Theta(i^{e_a}p_a^i)$ and $\varphi_b(i) = \Theta(i^{e_b}p_b^i)$, for some $e_a,e_b \geq 0$ and $p_a \neq p_b > 1$.

A morphism is called \emph{prolongable} on a letter $a\in A$ if $\varphi(a)=au$ with $u\in A^+$.
Then, for $i\geq 1$, $\varphi^i(a)=au\varphi(u)\cdots \varphi^{i-1}(u)$. In this case, the infinite family of finite words $\{a, \varphi(a), \ldots, \varphi^i(a), \ldots\}$
are prefixes of a unique infinite word denoted by $\varphi^{\infty}(a)$, that is called the \emph{word generated by the morphism $\varphi$}. It is also called \emph{purely morphic word}. 
Examples of infinite words generated by a morphism are the \emph{Thue-Morse word} $t=abbabaabbaababba\ldots$ generated by the Thue-Morse morphism $\tau$ and the \emph{Fibonacci word} $f=abaababaabaab\ldots$ generated by the Fibonacci morphism $\theta$.
More in general, an infinite word is called \emph{morphic} if is generated by applying a \emph{coding} (a $1$-uniform morphism from $A$ to a possibly different alphabet $\Sigma$) to a purely morphic word.


\section{BWT-Highly Compressible Morphisms}

In this section, we focus on morphisms that generate finite words $w$ on which the Burrows-Wheeler transform has a very effective action by significantly reducing the number of equal-letter runs. In particular, we show that some upper bounds depending on the factor complexity of the fixed point of the morphism can be derived.

\begin{definition}
A morphism $\varphi$ prolongable on $a \in A$ is $BWT$-highly compressible if  $\limsup_{i\to\infty}\rho(\varphi^i(a))=0$, where $\rho(\varphi^i(a))=\frac{r_{\bwt}(\varphi^i(a))}{r(\varphi^i(a))}$ is the $BWT$-clustering ratio of $\varphi^i(a)$.
\end{definition}




The factor complexity of purely morphic words has been studied  \cite{Pansiot_ICALP84}.

\begin{theorem}[\cite{Pansiot_ICALP84}]
\label{th:pansiot}
Let $x = \varphi^\infty(a)$ be an infinite aperiodic word and let $f_x$ be its factor complexity.
\begin{enumerate}
    \item If $\varphi$ is growing, then $f_x(n)$ is $\Theta(n)$, $\Theta(n \log \log n)$ or $\Theta(n \log n)$ if $\varphi$ is quasi-uniform, polynomially divergent or exponentially divergent,  respectively
    \item Let $\varphi$ be not-growing and let $B$ be the set of its bounded letters 
    \begin{enumerate}
        \item if $x$ has arbitrarily large factors of $B^*$ then $f_x(n) = \Theta(n^2)$ 
        \item if the factors of $B^*$ in $x$ have bounded length then $f_x(n)$ can be any of $\Theta(n)$, $\Theta(n \log \log n)$ or $\Theta(n \log n)$.
    \end{enumerate}
\end{enumerate}
\end{theorem} 

The following theorem shows that the factor complexity of some particular classes of morphisms is known \cite{DBLP:journals/tcs/EhrenfeuchtLR75}. 
Both Fibonacci morphism $\theta$ and Thue-Morse morphism $\tau$ are included in these classes. 

\begin{theorem}\label{th:u-morph_subwordcomplexity}
Let $x = \varphi^\infty(a)$ be an aperiodic infinite word that is the fixed point of the morphism $\varphi$.
If $\varphi$ is uniform or primitive, then $f_x(n) = \Theta(n)$.
\end{theorem}

The following two examples provide a $BWT$-highly compressible and a not $BWT$-highly compressible morphism, respectively. 

\begin{example}[$\Theta(n \log \log n)$ factor complexity]
Let us consider the binary morphism $\varphi$ defined as $\varphi(a)=abab$ and $\varphi(b)=bb$. 
In this case $\varphi_a(i)=(i+1)2^i$ and $\varphi_b(i)=2^i$.
Moreover, $r(\varphi^i(a))=2^{i+1}$ and $r_{\bwt}(\varphi^i(a))=2i$, for $i>2$. Hence, $\varphi$ is $BWT$-highly compressible.
\end{example}

\begin{example}[$\Theta(n \log n)$ factor complexity]
Let us consider the morphism $\psi$ defined as $\psi(a)=abc$, $\psi(b)=bb$ and $\psi(c)=ccc$. 
One can verify that $x=\psi^\infty(a)=abcb^2c^3b^4c^9\ldots$ and $\psi_a(i+1)=\psi_a(i)+2^i+3^i$. Moreover, $\rho(\psi^i(a))=\frac{r_{\bwt}(\psi^{i}(a))}{r(\psi^{i}(a))}=\frac{4i}{2i+1}>1$ for $i > 2$.
Hence, $\psi$ is not $BWT$-highly compressible.
\end{example}

The following proposition, only enunciated in \cite{NoiDCC}, gives an upper bound on the value $r_\bwt(\varphi^i(a))$, for some classes of morphisms prolongable on $a$. Such bounds depend on the factor complexity of the infinite word generated by $\varphi$. In the appendix a complete proof of the result is given.

\begin{proposition}
\label{prop:r_upper_bound_factorComplexity}
Let $x = \varphi^\infty(a)$ be an infinite aperiodic word. Then the following upper bounds for $r_\bwt(\varphi^i(a))$, $i \geq 1$, hold:
\begin{enumerate}
    \item if $f_x(n)=\Theta(n)$ then $r_{\bwt}(\varphi^i(a)) = \mathcal{O}(i)$.
    \item if $f_x(n)=\Theta(n \log \log n)$ then $r_{\bwt}(\varphi^i(a)) = \mathcal{O}(i \log i \log \log i)$.
    \item if $f_x(n)=\Theta(n \log n)$ then $r_{\bwt}(\varphi^i(a)) = \mathcal{O}(i^2 \log i)$.
\end{enumerate}
\end{proposition}

\begin{proof}
(Sketch) By \cite{rozenberg1980mathematical}, we know that the growth of any morphism is in $\mathcal{O}(\rho_a^i)$, for some $\rho_a>1$. 
Moreover, we can use an upper bound proved in \cite{KempaK20} to derive the thesis.
\end{proof}

Note that such upper bounds extend some known results. In fact, as shown in \cite{BrlekFMPR19}, $r_{\bwt}(\tau^i(a))=\Theta(i)$. 
We also remark that, since $n=|\tau^i(a)|=2^i$,  $r_{\bwt}(\tau^i(a))=\Theta(\log n)$. Furthermore, from results provided in \cite{BrlekFMPR19} and  \cite{MaReSc}, it can be deduced that $\tau$ and $\theta$ are $BWT$-highly compressible. However, the lower bounds can be quite different. In fact, as shown in \cite{MaReSc}, $r_{\bwt}(\theta^i(a))=\Theta(1)$. Actually, in the next section we show that, in case of binary alphabet, lower and upper bounds can be derived for some classes of morphisms.  

In the next example we show a class of morphisms $\varphi_k$, over an alphabet of size $k$, such that $r_{\bwt}(\varphi_k^i(a))=\Theta(n^{\frac{1}{k-1}})$, where $n=|\varphi_k^i(a))|$.

\begin{example}[$\Theta(n^2)$ factor complexity]
 Let us consider the morphism $$\varphi_k :\begin{array}{lll}
    a_1 & \mapsto & a_1 a_2\\
    a_2 & \mapsto & a_2 a_3 \\
    & \ldots & \\
    a_{k-1} & \mapsto & a_{k-1} a_k \\
    a_k & \mapsto & a_k\\
\end{array}$$   
One can verify that $n=|\varphi_k^i(a)| = \Theta(i^{k-1})$ \cite{ConstantinescuI07}.
Moreover, $r(\varphi_k^i(a)) = \Theta(i^{k-2})$ and $r_{\bwt}(\varphi_k^i(a)) = \Theta(i) = \Theta(n^{\frac{1}{k-1}})$.
Hence, $\varphi_k$ is $BWT$-highly compressible for any $k > 3$.
\label{ex:quadratic_k_morphism}
\end{example}

Here we introduce the notion of \emph{run-bounded} morphism in order to identify some classes of $BWT$-highly compressible morphisms.

\begin{definition}
Let $\varphi$ be a morphism such that $r(\varphi^i(a)) \leq K$, for any $i>0$, for some $a \in A$ and $K > 0$.
Then we say that $\varphi$ is \emph{run-bounded on $a$}.
\end{definition}

The following two propositions give bounds for $r(\varphi^i(a))$. Note that, since $\varphi$ is prolongable, $r(\varphi^i(a))$ is not decreasing. Corollary \ref{cor:primitive_morph} can be proved by using Propositions \ref{prop:exponential_growth_run} and \ref{prop:r_upper_bound_factorComplexity}.

\begin{proposition}
\label{prop:au_in_R}
    Let $\varphi$ be a morphism prolongable on $a \in A$ and let $\mathcal{R} = \{b\in A \mid \varphi$ is run-bounded on $b\}$.
    If $\varphi(a) = au$ with $u \in \mathcal{R}^+$ then $r(\varphi^i(a)) = \mathcal{O}(i)$.
\end{proposition}
\begin{proof}(Sketch)
    Recall that $\varphi^i(a) = a u \varphi(u) \varphi^2(u) \ldots \varphi^{i-1}(u)$.
    Since $\varphi$ is run-bounded on any symbol in $u$, 
    then $r(\varphi^i(u)) \leq K \cdot |u|$, for some $K>0$ and any $i>0$.
    Hence, we have that $r(\varphi^i(a)) \leq 1 + \sum_{j=0}^{i-1} r(\varphi^j(u))\leq (K\cdot|u|)\cdot i = \mathcal{O}(i)$.
\end{proof}

\begin{proposition}
\label{prop:exponential_growth_run}
    Let $\varphi$ be a morphism prolongable on $a \in A$, with $\varphi(a) = av$ with $|v|_b \geq 1$.
    If exists $t > 0$ such that $a$ occurs at least twice in $\varphi^t(a)$, then the growth of $r(\varphi^i(a))$ is exponential.
\end{proposition}
\begin{proof}
(Sketch) Let $t > 0$ be the smallest integer such that $a$ occurs at least twice in $\varphi^{t}(a)$.
    It is possible to prove by induction that, for any $j \geq 0$, $|\varphi^{j \cdot t}(a)|_a$ grows at least as $\Omega(2^j)$. Note that since $r(\varphi(a)) \geq 2$, then we can check that $r(\varphi^{j \cdot t + 1}(a)) \geq 2 \cdot |\varphi^{j \cdot t}(a)|_a = \Omega(2^{j +1})$.
    Since every $t$ steps the growth of the function is exponential, the overall growth is exponential too.
\end{proof}

\begin{corollary}\label{cor:primitive_morph}
Let $\varphi$ be a primitive morphism and prolongable on $a$ and let $n=|\varphi^i(a)|$, $i\geq 1$. It holds that $r_{\bwt}(\varphi^i(a))=\mathcal{O}(\log n)$ and $\lim_{i\to\infty}\rho(\varphi^i(a))=0$, therefore $\varphi$ is $BWT$-highly compressible.
\end{corollary}

The following example shows that, unlike primitive morphisms, uniform morphisms on generic alphabets can be not $BWT$-highly compressible. The situation becomes different if binary alphabets are considered, as shown in the next section.

\begin{example}\label{ex:3uniform_3letter}
    Let us consider the 3-uniform morphism $\eta$ defined as $\eta(a) = abc$, $\eta(b) = bbb$ and $\eta(c) = ccc$.
    It is easy to verify that in this case $r(\eta^i(a)) = r(abcbbbccc \ldots b^{3^{i-1}}c^{3^{i-1}}) = 2i + 1$ and $r_{\bwt}(\eta^i(a))) = 4i$, for any $i \geq 2$.
    Hence, $\rho(\eta^i(a)) = 2 - \frac{2}{2i + 1}$ and it is not $BWT$-highly compressible.
\end{example}

\section{Upper and Lower Bounds for $r_{\bwt}$}

In this section we show that lower and upper bounds for $r_{\bwt}$ of a word $w$ over a generic alphabet can be derived by considering the number of extensions of some bispecial circular factors of $w$.
By focusing on binary words over the ordered alphabet $A=\{a,b\}$ generated by a binary morphism $\mu$ prolongable on $a$, we give some new structural properties of their circular bispecial factors.
Such results allow us to derive logarithmic lower and upper bonds for some classes of binary morphisms.
Furthermore, we prove that for all the binary morphisms $\mu$ prolongable on $a$, except few cases, $\lim_{i\to\infty}\rho(\mu^i(a))=0$. Hence they are $BWT$-highly compressible.
Note that such results are independent of the order between the letters in $A$.

Let $u$ be a circular factor of a given word $w$ over a generic alphabet $A$. 
Inspired by the notation in \cite{Cassaigne97}, we denote by $e_r(u)=|\{x\in A \mid ux\in \mathcal{C}(w)\}|-1$ the number of \emph{right circular extensions of $u$ in $w$}, and by $e_\ell(u)=|\{x\in A \mid xu\in \mathcal{C}(w)\}|-1$ the number of \emph{left circular extensions of $u$ in $w$}. The bispecial circular factors of $w$ can be classified according to the number of their extensions. In particular, a circular factor $u$ is \emph{strictly bispecial} if $|\mathcal{C}(w)\cap AuA|=(e_r(u)+1)(e_\ell(u)+1)$, $u$ is \emph{weakly bispecial} if $|\mathcal{C}(w)\cap AuA|=\max \{e_r(u),e_\ell(u)\}+1$.
We denote by $SBS(w)$ and $WBS(w)$ the set of strictly and weakly bispecial circular factors of $w$, respectively. The following lemma holds. The proof is in the appendix.

\begin{lemma}
\label{le:r_bound_bispecial}
Let $w$ be a word over the alphabet $A$. Then, 
$$\sum_{u \in WBS(w)}\min\{e_l(u), e_r(u)\} +1 \leq r_{\bwt}(w) \leq \Sigma_{u\in BS(w)} e_r(u)+1.$$

\end{lemma}

From now on, we suppose that $A=\{a,b\}$. Given a binary morphism $\mu$, the notation $\mu\equiv(\alpha,\beta)$ means that $\mu(a)=\alpha$ and $\mu(b)=\beta$. 


\subsection{Combinatorial structure of binary morphisms}
In this subsection we give a combinatorial characterization of $\alpha$ and $\beta$, for several classes of binary morphisms $\mu\equiv (\alpha,\beta)$. Such a characterization depends on the factorial complexity of the fixed point $\mu^\infty(a)$.

Let $R_i = \{q_1 < q_2 < \ldots < q_h\}$ be the set of non-negative integers such that $ab^{q_j}a$ is a circular factor in $\mu^i(a)$, for some $1\leq j\leq h$. Denoted by $n_a=|\alpha|_a$, it is easy to verify that $|R_1|\leq n_a$, since $\mu(a) = ab^{t_1}ab^{t_2}a\cdots a b^{t_{n_a}}$ with $t_j \geq0$ for any $1 \leq j \leq n_a$.

The following proposition consider ultimately periodic purely morphic words and can be proved by using \cite[Corollary 3]{Pansiot86}.

\begin{proposition}\label{prop:struct_ult_per}
Let $x = \mu^\infty(a)$ be an infinite binary ultimately periodic word, where $\mu\equiv(\alpha,\beta)$. Then, one of the following cases must occur:
\begin{enumerate}
    \item $\alpha=\eta^\ell$  and $\beta=\eta^t$, for some $\ell, t\geq 1$;
    \item $\alpha=ab^k$ and $\beta=b^\ell$, for some $k\geq 1$, $\ell\geq 1$;
    \item $\alpha=(ab)^p a$ and $\beta=(ba)^q b$ for some $p,q \geq 1$; 
    \item $\alpha=(ab^p)^q a$ and $\beta=b$, for some $p\geq 1$, $q\geq 1$.
\end{enumerate}
\end{proposition}

 The two following lemmas give a combinatorial characterization of the sets $R_i$ for any non-primitive binary morphism. 
The proofs are in the appendix.

\begin{lemma}
\label{prop:R-auabk_b}
Let $\mu=(\alpha,\beta)$ be a non-growing morphism prolongable on $a$ and let $x = \mu^\infty(a)$ be its fixed point.  
\begin{enumerate}
    \item If $\alpha=auba^k$, $\beta=b$, for any $u\in A^*$, and $k\geq 1$ then $R_i=R_1$ for any $i\geq 1$;

    \item If $\alpha=auab^k$, $\beta=b$, for any $u\in A^*$, and $k\geq 1$ then  $R_{i} =  \bigcup_{h=1}^{i} \bigcup_{j=1}^{n_a - 1} \{t_j + (h-1)k\}\cup \{ik\}$ for any $i \geq 1$.
\end{enumerate}
\end{lemma}

\begin{lemma}
\label{prop:R-auabk_bl}
Let $\mu$ any growing non-primitive binary morphism. 
If $\mu^\infty(a)$ is ultimately periodic, then $\mu\equiv(ab^k,b^\ell)$, $k\geq 1,\ell> 1$ and $R_i=k\sum_{j=0}^{i-1}\ell^j$, $i\geq 1$.
Otherwise, it holds that $\mu \equiv (auab^k, b^\ell)$, $k\geq 0,\ell>1, u\in A^*$ and $R_{i} = \bigcup_{h=1}^{i} \bigcup_{j=1}^{n_a - 1} \{\frac{\ell^{h-1}((\ell-1)t_j +k) - k}{\ell-1}\} \cup \{k\frac{\ell^i-1}{\ell-1}\}$, for any $i \geq 1$.
\end{lemma}

The following two propositions give a combinatorial characterization of non-primitive morphisms generating aperiodic words. The proofs can be found in the appendix.

\begin{proposition}
\label{prop:struct_aperiodic_non-growing}
Let $\mu=(\alpha,\beta)$ be a non-growing morphism prolongable on $a$ and let $x = \mu^\infty(a)$ be its aperiodic fixed point.
Then, one of the following cases must occur:
\begin{enumerate}
    \item $\alpha=auba^k$, $\beta=b$, for some $u\in A^*$ and $k\geq1$, and $f_x(n) = \Theta(n)$;
    \item $\alpha=auab^k$, $\beta=b$, for any $u\in A^*$ and $k\geq1$, and $f_x(n) = \Theta(n^2)$.
\end{enumerate}
\end{proposition}

\begin{remark}
    Note that the morphism $\mu \equiv (auab^k,b)$, $k \geq 1$ always generates a fixed point with quadratic factor complexity for any $u \in A^*$. If $\mu \equiv(auba^k, b)$, $k\geq 1$, then we have to distinguish two cases: if $k>1$, then $\mu$ generates a fixed point with linear factor complexity for any $u \in A^*$; if $k=1$, by Proposition \ref{prop:struct_ult_per} (case $4$), $u\neq (b^pa)^qb^{p-1}$, $p,q\geq 1$, otherwise the fixed point is ultimately periodic. 
    
\end{remark}

\begin{proposition}
\label{prop:struct_aperiodic_growing}
Let $\mu=(\alpha,\beta)$ be a growing non-primitive morphism prolongable on $a$ and let $x = \mu^\infty(a)$ be its aperiodic fixed point.
Then, $\mu \equiv (av, b^\ell)$, for some $\ell \geq 2$ and $v\in A^+$ such that $|v|_a, |v|_b \geq 1$.
Moreover, let $n_a = |av|_a$.
Then, it holds that: 
\begin{enumerate}
    \item $n_a < \ell$ iff $f_x(n) = \Theta(n)$;
    \item $n_a = \ell$  iff $f_x(n) = \Theta(n \log \log n)$;
    \item $n_a > \ell$ iff $f_x(n) = \Theta(n \log n)$.
\end{enumerate}
\end{proposition}

The following proposition shows that for any binary non-primitive morphism, except for the case of non-growing morphisms with linear factorial complexity of the fixed point (case $1$ of Lemma \ref{prop:R-auabk_b}), the size of the set $R_i$ of non-negative integers $j$ such that $ab^ja$ is a circular factor of $\mu^i(a)$ grows linearly with $i$.

\begin{proposition}
\label{prop:linear_R_i}
Let $\mu\equiv (auab^k, b^\ell)$ a binary morphism with $k\geq 0$, $\ell\geq 1$ and $k+\ell>1$ and aperiodic fixed point. Then $|R_i|=\Theta(i)$.
\end{proposition}

\subsection{Logarithmic bounds for $r_{\bwt}$ in case of binary morphisms}
In this subsection we prove that if $w$ is a finite word of length $n$ generated by iterating a binary morphism, then $r_{\bwt}(w)=\mathcal{O}(\log n)$. Moreover, we identify some classes of binary morphisms for which $\Omega(\log n)$ is a lower bound for $r_{\bwt}$. From Proposition \ref{prop:struct_ult_per}, one can easily derive that in case of ultimately periodic words $r_{bwt}$ is $\Theta(1)$. In case of a primitive morphism $\mu$ the upper bound $\mathcal{O}(\log n)$ can be deduced by Proposition \ref{prop:r_upper_bound_factorComplexity} and Theorem \ref{th:u-morph_subwordcomplexity}, by using the fact that, in this case, $\mu_i(a)$ is exponential \cite{rozenberg1980mathematical}. Hence, here we can suppose that the morphism $\mu\equiv(\alpha,\beta)$ is not primitive and $\mu^\infty(a)$ is aperiodic. 

The following lemmas give a structural characterization of the bispecial circular factors of the words generated by iterating a binary morphism. 
In particular, Lemma \ref{lemma:b^(k)f(v)-general} shows how to construct bispecial circular factors of $\mu^{i+1}(a)$ starting from the bispecial circular factors of $\mu^{i}(a)$, $i\geq 1$. In Lemma \ref{le:structural_general} we prove that all bispecial circular factors can be constructed by starting from the bispecial circular factors of a finite set of words depending of the images of $\mu$ on the letters of the alphabet. The proofs are in the appendix.

\begin{lemma}
\label{lemma:b^(k)f(v)-general}
Let $\mu \equiv(auab^k, b^\ell)$ be a binary morphism, for some $ k\geq 0, \ell \geq 1$.
If $v$ is a circular bispecial factor of $\mu^{i}(a)$, $i\geq 1$, then $w = b^k\mu(v)$ is a bispecial circular factor of $\mu^{i+1}(a)$.
\end{lemma}

\begin{lemma}
\label{le:structural_general}
Let $\mu \equiv (\alpha, \beta) = (auab^k, b^\ell)$ for some $k\geq 0$, $\ell \geq 1$, 
let $m$ be the length of the longest equal-letter run of $b$'s that occurs in $auab^k$,
and let $M=\max\{\lfloor \frac{m-(\ell+1)k}{\ell^2} \rfloor, 0\}$.
Then, any circular bispecial factor $w$ of $\mu^{i+1}(a)$, $i\geq 1$, either appears as a circular factor in $ \bigcup_{j=0}^{M}\{\mu(\alpha)\mu(\beta)^{j}\mu(\alpha)\}$ or $w=b^k \mu(v)$, for some circular bispecial factor $v$ in $\mu^{i}(a)$ (or $w = b^h$, for some $h\geq 1$, when $k\geq 0$ and $\ell>1$).
\end{lemma}


\begin{theorem}\label{prop:upper_bound}
Let $\mu\equiv(auab^k,b^\ell)$ be a non-primitive morphism with $k\geq0$, $\ell \geq 1$ with aperiodic fixed point.
Then $r_{bwt}(\mu^i(a))= \mathcal{O}(i)$.    
\end{theorem}
\begin{proof}
(Sketch) The case $k=0$ and $\ell=1$ follows from Proposition \ref{prop:struct_aperiodic_non-growing} and Proposition \ref{prop:r_upper_bound_factorComplexity}. In case of binary alphabet, Lemma \ref{le:r_bound_bispecial} implies that, for any $i\geq 1$, $r_{\bwt}(\mu^i(a))\leq |BS(\mu^i(a))|+1$. 
If $k \geq 1$ and $\ell = 1$, then $|BS(\mu^i(a))| = \mathcal{O}(i)$ by using Lemma \ref{le:structural_general}. If $k \geq 0$ and $\ell > 1$, $|BS(\mu^i(a))|$ grows exponentially since $BS$ contains the subset $BS_b(\mu^i(a)) = \{b^h \mid b^h, b^{h+1}\in \mathcal{C}(\mu^i(a))\}$ by Lemma \ref{le:structural_general}. However, only the elements of a subset of $BS_b(\mu^i(a))$ size at most $2|R_i|$ produce an increase of $1$ for  $r_{\bwt}(\mu^i(a))$. The thesis follows from Proposition \ref{prop:linear_R_i}.
\end{proof}

\begin{theorem}
\label{prop:lowerbound_r_Rn}
Let $\mu\equiv(auab^k,b^\ell)$ be a non-primitive morphism with $k\geq0$, $\ell \geq 1$ and $x=\mu^\infty(a)$ is aperiodic.
Then $r_{\bwt}(\mu^i(a)) = \Omega(|R_i|)$.
Moreover, when $\mu$ is not growing with $f_x(n)=\Theta(n^2)$ or $\mu$ is growing, then $r_{\bwt}(\mu^i(a)) = \Omega(i)$. 
\end{theorem}
\begin{proof}
(Sketch) Let $R_i=\{q_1<q_2<\ldots<q_{|R_i|}\}$ the set of non-negative integers such that $ab^{q_j}a\in \mathcal{C}(\mu^i(a))$. For any $q_j\in R_i$, we can consider the block $X_j$ of lexicographically sorted conjugates starting with $b^{q_j}a$. Among the corresponding characters in $\bwt(\mu^i(a))$, at least one occurrence of the letter $a$ is included. For any $0\leq j\leq \frac{\lceil|R_i|-1\rceil}{2}$, let us consider the blocks $X_{2j+1}$ and $X_{2j+3}$. Since $q_{2j+1}<q_{2j+3}-1$, there are at least $|X_{2j+3}|$ lexicographically sorted conjugates starting with $b^{q_{2j+3}-1}a$ and ending with $b$. The second part of the thesis is proved by using Propositions \ref{prop:struct_aperiodic_non-growing}, \ref{prop:struct_aperiodic_growing} and \ref{prop:linear_R_i}. 
\end{proof}

\begin{remark}
    Note that Theorem \ref{prop:upper_bound} and Theorem \ref{prop:lowerbound_r_Rn} also hold when $b<a$. 
\end{remark}

The following lemma and corollary allow us to states that, for morphisms focused in this subsection, $i=\Theta(\log n)$, where $n=\mu_a(i)=|\mu^i(a)|$.

\begin{lemma}\label{lb}
Let  $\mu \equiv (\alpha,\beta)$ a binary morphism prolongable on $a$. Let $n_a=|\alpha|_a, n_b=|\alpha|_b, m_a=|\beta|_a, m_b=|\beta|_b$ and $\alpha_i=|\mu^{i-1}(a)|_a, \beta_i=|\mu^{i-1}(a)|_b$. It holds that:
$$r(\mu^i(a))\geq \alpha_i r(\alpha) + \beta_i r(\beta) - |\mu^{i-1}(a)|+1
$$
with 
$$
\begin{array}{l}
\alpha_i = n_a \alpha_{i-1} + m_an_b\alpha_{i-2}+m_am_bn_b\alpha_{i-3}+m_am_b^2n_b\alpha_{i-4}+\cdots+m_am_b^{i-3}n_b \\
\beta_i=m_b\beta_{i-1}+m_an_b\beta_{i-2}+m_an_an_b\beta_{i-3}+m_an_a^2n_b\beta_{i-4}+\cdots+m_an_a^{i-4}n_b\beta_2+ n_a^{i-2}n_b.
\end{array}
$$
\end{lemma}

From Lemma \ref{lb} and Proposition  \ref{prop:exponential_growth_run}, the following corollary follows. 

\begin{corollary}
Let $\mu=(\alpha,\beta)$ a binary morphism  prolongable on $a$. 
Then, the growth of $\mu_a(i)$ is exponential except when $\alpha=ab^p$, with $p\geq 1$, and $\beta=b$, where $\mu_a(i)= \Theta(i)$. 
\end{corollary}

The goal of the following result is to evaluate the $BWT$-clustering ratio of the finite words generated by iterating a binary morphism. The proof can be derived from Lemma \ref{lb},  Proposition \ref{prop:struct_ult_per}, Corollary \ref{cor:primitive_morph}, and Theorem \ref{prop:upper_bound}.

\begin{theorem}\label{th:highlyBWTcompres}
Let $\mu\equiv(\alpha,\beta)$ be a binary morphism prolongable on $a$ such that $\mu\not \equiv(ab^m,b^n)$ 
for any $m\geq 1$, $n\geq 1$. Then $\lim_{i\to\infty}\rho(\mu^i(a))=0$, consequently $\mu$ is $BWT$-highly compressible.
\end{theorem}

\section{Conclusions and further work}
In this paper, we have studied the number $r_{\bwt}(w)$ of equal-letter runs  produced by the $BWT$, when $w=\mu^i(a)$ of length $n$ is the binary word generated after the $i$-th iteration of a morphism $\mu$ prolongable on the letter $a$ with an aperiodic fixed point $x=\mu^\infty(a)$. We have proved that $r_{\bwt}(w)$ is $\Theta(\log n)$ when $\mu$ is a non-primitive growing morphism or a non-primitive not-growing morphism such that  $f_x(n)=\Theta(n^2)$. It is still open the problem to characterize the primitive morphisms such that $r_{\bwt}(w)$ is $\Omega(\log n)$. This could allow a tight lower bound to be deduced even for non-primitive not growing morphisms such that $f_x(n)=\Theta(n)$.  Moreover, we are interested to extend these bounds also for purely morphic finite words on larger alphabets, and also for generic morphic finite words.

%
%

\bibliographystyle{splncs04}
\bibliography{biblio}

\begin{thebibliography}{10}
\providecommand{\url}[1]{\texttt{#1}}
\providecommand{\urlprefix}{URL }
\providecommand{\doi}[1]{https://doi.org/#1}

\bibitem{BrlekFMPR19}
Brlek, S., Frosini, A., Mancini, I., Pergola, E., Rinaldi, S.: {Burrows-Wheeler
  Transform of Words Defined by Morphisms}. In: {IWOCA}. Lect. Notes Comput.
  Sci., vol. 11638, pp. 393--404. Springer (2019)

\bibitem{BW94}
Burrows, M., Wheeler, D.J.: A block-sorting lossless data compression
  algorithm. Tech. rep., DIGITAL System Research Center (1994)

\bibitem{Cassaigne97}
Cassaigne, J.: Complexity and special factors. (complexité et facteurs
  spéciaux.). Bulletin of the Belgian Mathematical Society - Simon Stevin
  \textbf{4}(1),  67--88 (1997)

\bibitem{ChristiansenEKN21}
Christiansen, A.R., Ettienne, M.B., Kociumaka, T., Navarro, G., Prezza, N.:
  Optimal-time dictionary-compressed indexes. {ACM} Trans. Algorithms
  \textbf{17}(1),  8:1--8:39 (2021)

\bibitem{ConstantinescuI07}
Constantinescu, S., Ilie, L.: {The Lempel--Ziv Complexity of Fixed Points of
  Morphisms}. {SIAM} J. Discret. Math.  \textbf{21}(2),  466--481 (2007)

\bibitem{DBLP:journals/tcs/EhrenfeuchtLR75}
Ehrenfeucht, A., Lee, K.P., Rozenberg, G.: Subword complexities of various
  classes of deterministic developmental languages without interactions. Theor.
  Comput. Sci.  \textbf{1}(1),  59--75 (1975).
  \doi{10.1016/0304-3975(75)90012-2},
  \url{https://doi.org/10.1016/0304-3975(75)90012-2}

\bibitem{cj/Fenwick96}
Fenwick, P.: The {B}urrows-{W}heeler transform for block sorting text
  compression: Principles and improvements. Comput. J.  \textbf{39}(9),
  731--740 (1996)

\bibitem{Ferenczi_Zamboni2013}
{Ferenczi}, S., {Zamboni}, L.Q.: {Clustering Words and Interval Exchanges}.
  Journal of Integer Sequences  \textbf{16}(2),  Article 13.2.1 (2013)

\bibitem{Ferragina:2005}
Ferragina, P., Manzini, G.: Indexing compressed text. J. ACM  \textbf{52},
  552--581 (2005)

\bibitem{NoiDCC}
Frosini, A., Mancini, I., Rinaldi, S., Romana, G., Sciortino, M.:
  Burrows-wheeler transform on purely morphic words. In: {DCC 2022}. pp.~1--1.
  {IEEE} (To appear)

\bibitem{GagieNP20}
Gagie, T., Navarro, G., Prezza, N.: {Fully Functional Suffix Trees and Optimal
  Text Searching in BWT-Runs Bounded Space}. J. {ACM}  \textbf{67}(1),
  2:1--2:54 (2020)

\bibitem{GiulianiILPST21}
Giuliani, S., Inenaga, S., Lipt{\'{a}}k, Z., Prezza, N., Sciortino, M.,
  Toffanello, A.: Novel results on the number of runs of the
  burrows-wheeler-transform. In: {SOFSEM}. Lect. Notes Comput. Sci., vol.
  12607, pp. 249--262. Springer (2021)

\bibitem{KempaK20}
Kempa, D., Kociumaka, T.: Resolution of the burrows-wheeler transform
  conjecture. In: {FOCS}. pp. 1002--1013. {IEEE} (2020)

\bibitem{KempaP18}
Kempa, D., Prezza, N.: At the roots of dictionary compression: string
  attractors. In: {STOC}. pp. 827--840. {ACM} (2018)

\bibitem{KociumakaNP20}
Kociumaka, T., Navarro, G., Prezza, N.: Towards a definitive measure of
  repetitiveness. In: {LATIN}. Lect. Notes Comput. Sci., vol. 12118, pp.
  207--219. Springer (2020)

\bibitem{MantaciRRS_words2017}
Mantaci, S., Restivo, A., Rosone, G., Sciortino, M.: {Burrows-Wheeler Transform
  and Run-Length Enconding}. In: {WORDS}. Lect. Notes Comput. Sci., vol. 10432,
  pp. 228--239. Springer (2017)

\bibitem{MANTACI_TCS2017}
Mantaci, S., Restivo, A., Rosone, G., Sciortino, M., Versari, L.: {Measuring
  the clustering effect of BWT via RLE}. Theoret. Comput. Sci.  \textbf{698},
  79 -- 87 (2017)

\bibitem{MaReSc}
Mantaci, S., Restivo, A., Sciortino, M.: Burrows-{W}heeler transform and
  {S}turmian words. Inform. Process. Lett.  \textbf{86},  241--246 (2003)

\bibitem{Navarro21a}
Navarro, G.: Indexing highly repetitive string collections, part {I:}
  repetitiveness measures. {ACM} Comput. Surv.  \textbf{54}(2),  29:1--29:31
  (2021)

\bibitem{NavarroU21}
Navarro, G., Urbina, C.: On stricter reachable repetitiveness measures. In:
  {SPIRE}. Lect. Notes Comput. Sci., vol. 12944, pp. 193--206. Springer (2021)

\bibitem{Pansiot_ICALP84}
Pansiot, J.: Complexit{\'{e}} des facteurs des mots infinis engendr{\'{e}}s par
  morphimes it{\'{e}}r{\'{e}}s. In: {ICALP}. Lect. Notes Comput. Sci.,
  vol.~172, pp. 380--389. Springer (1984)

\bibitem{Pansiot86}
Pansiot, J.J.: Decidability of periodicity for infinite words. RAIRO - Theor.
  Inform. Appl.  \textbf{20}(1),  43--46 (1986)

\bibitem{RestivoRosoneTCS2009}
Restivo, A., Rosone, G.: {B}urrows-{W}heeler transform and palindromic
  richness. Theoret. Comput. Sci.  \textbf{410}(30-32),  3018 -- 3026 (2009)

\bibitem{rozenberg1980mathematical}
Rozenberg, G., Salomaa, A.: The Mathematical Theory of L Systems. Elsevier
  Science (1980), \url{https://books.google.it/books?id=0rr0BJxDKSwC}

\bibitem{SewardBzip2}
Seward, J.: The {\sc bzip2} home page (2006), {\tt http://www.bzip.org}

\bibitem{shallit_shaeffer2020}
Shallit, J., Shaeffer, L.: String attractors for automatic sequences. CoRR
  \textbf{abs/2012.06840} (2020)

\bibitem{puglisiSimpson2008}
Simpson, J., Puglisi, S.J.: Words with simple {B}urrows-{W}heeler transforms.
  Electronic Journal of Combinatorics  \textbf{15} (article R83, 2008)

\end{thebibliography}

\newpage
\section*{Appendix}

It has been proved that an upper bound on the number of equal-letter runs of $BWT$ can be derived by using the measure $\delta$ \cite{KociumakaNP20,ChristiansenEKN21}, defined as follows: given a finite word $w\in A^*$, $\delta(w) = \max\{f_w(k)/k, 1\leq k\leq |w|\}$ .

\begin{theorem}[\cite{KempaK20}]
\label{th:r_as_delta}
Let $w$ be a word over the alphabet $A$. Then $r_{\bwt}(w) = \mathcal{O}(\delta(w) \log \delta(w) \cdot \max\{1, \log \frac{|w|}{\delta(w) \log \delta(w)}\})$.
\end{theorem}

\setcounter{myCounter}{\getrefnumber{prop:r_upper_bound_factorComplexity}}
\begin{proofofProp}
Let $w_i = \varphi^i(a)$ be the $i$-th iterate of the morphism $\varphi$ on the symbol $a$.
By \cite{rozenberg1980mathematical}, we know that the growth of any morphism is $\mathcal{O}(\rho_a^i)$, for some $\rho_a>1$.
Moreover, we can use Theorem \ref{th:r_as_delta} to derive a bound on $r_\bwt(w_i)$, depending on the behaviour of $\delta(w_i)$ in the following three cases:

\textit{1.} $f_x(n) = \Theta(n)$: We have that $f_{w_i}(k) \leq f_{x}(k) \leq t \cdot k$ for any $i\geq 0$, $k > 0$ and for some $t > 0$.
    This implies that $\delta(w_i) \leq t$, for any $i \geq 0$.
    Moreover, $\delta(w_i) \in \mathcal{O}(1)$ implies that there exists a constant $t'$ such that $\delta(w_i) \log \delta(w_i) \leq t'$.
    Recall that $|w_i| = O(\rho_a^i)$.
    Since $t'$ and $\rho_a$ have constant values and assuming $\log\frac{|w_i|}{\delta(w_i) \log \delta(w_i)} > 1$, there exist constant values $c_1, c_2 >0 $ such that $\delta(w_i) \log \delta(w_i) \cdot \max\{1,\log\frac{|w_i|}{\delta(w_i) \log \delta(w_i)}\} \leq  t' (\log|w_i| - \log t') \leq t' (c_1 \cdot i\log \rho_a - \log t') \leq c_2 \cdot i$, and by Theorem \ref{th:r_as_delta}, $r_\bwt(w_i) = \mathcal{O}(i)$.
    
    \textit{2.} $f_w(n) = \Theta(n \log \log n)$: 
    Since $f_{w_i}(k) \leq f_x(k) \leq t \cdot k \log \log k$ for any $i \geq 0$, $k >0$ and for some $t > 0$, we have $\delta(w_i) \leq t \cdot \log \log |w_i|$.
    Hence, $\delta(w_i) \log \delta(w_i) \cdot \max \{1, \log \frac{|w_i|}{\delta(w_i) \log \delta(w_i)}\} \leq t' \log \log |w_i| \log \log \log |w_i| \cdot \log |w_i|$ for some $t'>0$, and by Theorem \ref{th:r_as_delta} $r_\bwt(w_i) = \mathcal{O}(i \log i \log \log i)$. 
    
    \textit{3.} $f_w(n) = \Theta(n \log n)$: Analogously to the previous case, $f_{w_i}(k) \leq f_x(k) \leq t \cdot k \log k$, for any $i \geq 0$, $k >0$ and for some $t > 0$, and therefore $\delta(w_i) \leq t \cdot \log |w_i|$.
    Hence, $\delta(w_i) \log \delta(w_i) \cdot \max \{1, \log \frac{|w_i|}{\delta(w_i) \log \delta(w_i)}\} \leq t' \log |w_i| \log \log |w_i| \cdot \log |w_i|$  for some $t'>0$, and by Theorem \ref{th:r_as_delta} $r_{\bwt}(w_i)= \mathcal{O}(i^2 \log i)$.
    \qed
\end{proofofProp}

\setcounter{myCounter}{\getrefnumber{prop:exponential_growth_run}}
\begin{proofofProp}
Let $t > 0$ be the smallest integer such that $a$ occurs twice in $\varphi^{t}(a)$.
    We prove by induction that $|\varphi^{j \cdot t}(a)|_a$ grows as $\Omega(2^j)$. 
    By definition of $t$, the case $j = 1$ is trivial.
    For the inductive step, assume it holds that $|\varphi^{j \cdot t}(a)|_a = \Omega(2^j)$.
    Let $a_k$ be the $k$th occurrence of $a$ in $\varphi^{j \cdot t }(a)$.
    This implies that $\varphi^{(j+1) t}(a) = \varphi^t(\varphi^{j t}(a)) = \varphi^t(a_1 \ldots a_2 \ldots a_{2^j} \ldots) = \varphi^t(a_1) \ldots \varphi^t(a_2) \ldots \varphi^t(a_{2^j}) \ldots$.
    Since $\varphi^t(a_i)$ produces at least 2 $a$'s for each $i$ and the morphism is prolongable, then $|\varphi^{(j+1) \cdot t}(a)|_a \geq 2 |\varphi^{j \cdot t}(a)|_a \in \Omega(2^{j+1})$.
    Note that since $r(\varphi(a)) \geq 2$, then $r(\varphi^{j \cdot t + 1}(a)) \geq 2  |\varphi^{j \cdot t}(a)|_a = \Omega(2^{j +1})$ (this holds even if we assume that the last character of $\varphi(a)$ is an $a$ as well, since this would imply that $r(\varphi(a)) \geq 3$).
    Since every $t$ steps the growth of the function is exponential, the overall growth is exponential too.
\qed
\end{proofofProp}

\setcounter{myCounter}{\getrefnumber{le:r_bound_bispecial}}
\begin{proofofLemma}
    Let $y = \bwt(w)$ and let $S = \{p_1 < p_2 < \ldots < p_{r-1}\}$ be the set of positions in $\bwt(w)$ of the last character of the first $r_\bwt(w)-1$ equal-letter runs in the $BWT$. This means that the equal-letter runs of $y$ are $y_1\cdots y_{p_1}, y_{p_1+1}\cdots y_{p_2}, \ldots, y_{p_{r-2}+1} \cdots y_{p_{r-1}}, y_{p_r}\cdots y_n$.
    It follows that $y_{p_j} \neq y_{p_j+1}$ for any $1 \leq j \leq r-1$.
    Let moreover $w_j$ be the $j$th conjugate of $w$ in lexicographical order.
    We can then define the set $W = \{C_{p_1}, C_{p_2}, \ldots, C_{p_{r-1}}\}$, that is the set of lexicographically sorted cyclic rotations of $w$ corresponding to the positions in $S$.
    We observe that every $u_i = lcp(C_{p_i}, C_{p_i + 1})$ (where $lcp$ denotes the longest common prefix between the two words) is a bispecial circular factor, since $C_{p_i} = u_i a v'$ and  $C_{p_i + 1} = u_i b v''$, for some $a<b \in A$ and $v', v'' \in A^*$. In fact, either $v'$ and $v''$ are both empty or they end with different letters. 
    Moreover, each factor $u_i$ can correspond to at most $e_r(u_i)$ distinct positions in $S$. 
    By contradiction, let us suppose that there exist $p_{i_1} < p_{i_2} < \ldots < p_{i_{e_r(u_i)+1}} \in S$ such that $lcp(C_{p_{i_j}}, C_{p_{i_j}+1}) = u_i$ for each $j \in [1..e_r(u)+1]$. 
    Then we should find in $W$ at least two distinct conjugates $C_{p_{i_j}}$ and $C_{p_{i_{j'}}}$, for some $j < j'$ that both start $u_i a$, for some $a \in A$.
    Consider $C_{p_{i_j} + 1}$, that is the first conjugate following $C_{p_{i_j}}$ in lexicographical order.
    By definition, $C_{p_{i_j} + 1} = u_i bv$, for some $b>a$ and $v \in A^*$.
    Then, we would have $C_{p_{i_j}} = u_i a v' < C_{p_{i_{j'}}} = u_i a v'' < C_{p_{i_j} + 1} = u_i b v$, for some $v', v'' \in A^*$, but this is impossible since $C_{p_{i_j}}$ and $C_{p_{i_j} + 1}$ are consecutive and $C_{p_{i_j} + 1} \neq C_{p_{i_{j'}}}$,  contradiction.

    Let $e(u) = \max\{e_l(u), e_r(u)\}$.
    By definition, for any $u \in WBS(w)$, there are exactly $e(u) + 1$ distinct circular factors of $w$ of the type $aua'$, for some $a,a' \in A$.
    Let $S_u = \{q_1 < q_2 < \ldots < q_{e_r(u)}\}$ be the set of positions in the $BWT$ such that $C_{q_j} = uav_1$ and $C_{q_j + 1} = ubv_2$, for any $j \in [1..e_r(u)]$, some $a,b \in A$ such that $a<b$ and some $v_1, v_2 \in A^*$.
    If $e_l(u) \leq e_r(u)$, then $e_l(ua) = 0$ (i.e. it exists only one left extension of $ua$) for any $a \in A$.
    It follows that there are at least $e_l(u)$ cyclic rotations such that $C_{q_i} = uav_1$ and $C_{q_i+1} = ubv_2$ end with distinct characters, i.e. there are at least $e_l(u)$ changes of letters in $y$ that uniquely corresponds to the weak bispecial factor $u$.
    On the other hand, if $e_l(u) > e_r(u)$,  then $e_r(au) = 0$ (i.e. it exists only one right extension of $au$) for any $a \in A$. 
    It follows that if there exist two consecutive cyclic rotations $C_q$ and $C_{q+1}$ such that $C_q = uav_1'$ such that $C_{q+1} = ubv_2'$, then either $v_{1}'=v_{2}'=\varepsilon$ or they end with different characters, i.e. there are at least $e_r(u)$ changes of letters in $y$ that correspond to $u$ and the thesis follows.
\qed
\end{proofofLemma}

Before proving the next proposition, we need the following definitions \cite{rozenberg1980mathematical}.
A morphism $\varphi$ is called \emph{simplifiable} if there exist an alphabet $A'$ with $|A'| < |A|$, and two morphisms $\chi : A^* \to A'^*$ and $\chi' : A'^* \to A^*$ such that $\varphi(w) = \chi'(\chi(w))$ for any $w \in A^*$, otherwise $\varphi$ is called \emph{elementary}.

\setcounter{myCounter}{\getrefnumber{prop:struct_ult_per}}
\begin{proofofProp}
Let $\mu$ be a simplifiable morphism.
This means that exist $\chi : \{a,b\}^* \to \{a\}^*$ and $\chi' : \{a\}^* \to \{a,b\}^*$ such that $\mu(w) = \chi'(\chi(w))$ for any $w \in A^*$.
More into details, we have that $\chi(a) = a^t$ and $\chi(b) = a^{t'}$ for some $t,t'>0$ and $\chi'(a) = \eta$ for some $\eta \in \{a,b\}^*$.
It is easy to see that the fixed point of any of these morphisms is of the type $x = \eta^\infty$ that is ultimately periodic, and therefore $\mu(a) = \chi'(\chi(a)) = \chi'(a^t) = \eta^t$ and $\mu(b) = \chi'(\chi(b)) = \chi'(a^{t'}) = \eta^{t'}$ (that is the case {\it 1.}).

Let us suppose that $\mu$ is elementary and let $G$ and $B$ be the set of growing and bounded letters respectively.
Since we assume the morphism is prolongable on $a$, we have that $a \in G$.
By \cite[Corollary 3]{Pansiot86}, if the fixed point of the morphism contains only one growing letter, then it is ultimately periodic.
Note that this can occur only when $|\mu(a)|_a = 1$ and $b\in B$, that is $\mu\equiv (ab^k,b)$ for any $k \geq 1$ (case {\it 2.} for $ \ell = 1$).

Otherwise, let $P_x = v_0 c_1 v_1 c_2 v_2 \ldots v_{h-1} c_h v_h \ldots v_{l-1} c_l v_{l} c_h$, with $c_j \in G$ for any $1 \leq j \leq l$ and $v_{j'} \in B^*$ for any $0 \leq j' \leq l$, be the shortest prefix of $x$ that contains two occurrences of the same letter in $G$.
It follows that $c_j \neq c_j'$ for any $1 \leq j \neq j' \leq l$. 
By \cite[Corollary 3]{Pansiot86}, if for any prefix $c_j v$ of any $c_j v_j$ factor of $P_x$ it holds that $e_r(c_j v) = 0$, then $x$ is ultimately periodic.

If $b \in B$, then $P_x = a b^p a$ for some $p \geq 0$.
It follows that in order to have an ultimately periodic word, $e_r(ab^{p'}) = 0$ for any $0 \leq p' \leq p$.
We can see that this is equivalent to say that $|R_i| = 1$ for any $i \geq 1$.
In fact, if $|R_i| \geq 2$, then $ab^ja, ab^{j'a} \in \mathcal{C}(\mu^i(a))$ with $0 \leq j < j'$ and $e_r(ab^j) = 1$ (where $ab^j$ is both prefix of $ab^j$ and $ab^{j'}$) and $x$ would not be ultimately periodic.
It follows that the fixed point must be $x = (ab^p)^\infty$.
Since $\mu(a)$ is a prefix of $x$, we have that $\mu(a) = (ab^p)^qab^{p'}$, for some $p,q \geq 1$ and $0 \leq p' \leq p$ (notice that if $p = 0$ or $q=0$, then $x=a^\infty$ or $x = ab^\infty$ respectively).
However, we can see that $\mu^2(a) = \mu((ab^p)^q ab^{p'}) = (((ab^p)^q ab^{p + p'})^q (ab^p)^q ab^{p+p'})$.
Note that if $p' > 0$ then $p, p+p' \in R_i$ and therefore $x$ can not be ultimately periodic.
Therefore, $p' = 0$ and $\mu \equiv ((ab^p)^qa,b)$ is ultimately periodic (case {\it 4.}).

If $b \in G$ (and therefore $B = \emptyset$), then ($i$) $P_x = abb$ or ($ii$) $P_x = aba$.
Since it must hold that $e_r(a) = e_r(b) = 0$, then either $x = ab^\infty$ or $x = (ab)^\infty$ (($i$) and ($ii$) respectively).
For case ($i$) we can check that, since $\mu(a)$ is prefix of the fixed point, $\mu \equiv (ab^k, b^\ell)$ for some $k \geq 1$ and $\ell > 1$ (case {\it 4.} for $\ell > 1$).
For case ($ii$), note that by definition $x = \mu(x) = (\alpha\beta)^\infty$.
Therefore, either $\mu \equiv ((ab)^t, (ab)^\ell)$ for some $t,\ell \geq 1$ (that is case {\it 1.} again), or $\mu \equiv ((ab)^pa, (ba)^qb)$ for some $p,q \geq 1$ (case {\it 3.}) and the thesis follows. \qed
\end{proofofProp}

\setcounter{myCounter}{\getrefnumber{prop:R-auabk_b}}
\begin{proofofLemma}
Suppose $\mu(a) = auab^k$ with $k \geq 1$.
We can see that for $R_1 = \bigcup_{j=1}^{n_a}\{t_j\}$ the thesis holds.
By induction, suppose it holds for $R_i$.
We can see that all $a b^{q} a \in \mathcal{C}(\mu^{i+1}(a))$, for some $q\in R_{i+1}$, either are factors of $\mu(a) = auab^k$ (that is 
$q = t_j$ for some $1 \leq j < n_a$) or there exists $q' = q-k$ such that $q' \in R_i$ and $\mu(ab^{q'}a) = au\cdot ab^{q}a\cdot uab^k$.
Hence, $R_{i+1} = \bigcup_{j=1}^{n_a-1} \{t_j\} \cup \bigcup_{q' \in R_i} \{q'+k\} = \bigcup_{h=1}^{i+1} \bigcup_{j=1}^{n_a - 1} \{t_j + (h-1)k\} \cup \{(i+1)k\}$.
For case {\it 1}, note that the proof holds also in this case but with $t_{n_a} = k = 0$, hence we obtain $R_{i+1} =  \bigcup_{j=1}^{n_a} \{t_j\} = R_1$ and the thesis follows.
\qed
\end{proofofLemma}

\setcounter{myCounter}{\getrefnumber{prop:R-auabk_bl}}
\begin{proofofLemma}
The proof is analogous to that of Lemma \ref{prop:R-auabk_b}.
\qed
\end{proofofLemma}

In order to give a proof for Proposition \ref{prop:struct_aperiodic_non-growing} we need to prove the following lemma.

\begin{lemma}
\label{le:non-primitive-aperiodic}
Let $\mu$ be a binary non-primitive morphism prolongable on $a$ with fixed point $x = \mu^\infty(a)$ aperiodic.
Then $\mu \equiv (au,b^\ell)$ for some $u\in \Sigma^+$ such that $|u|_a, |u|_b\geq 1$ and for some $\ell \geq 1$.
\end{lemma}

\begin{proof}
In order to be prolongable on $a$, we must have $\mu(a) = au$ for some  $u \in \Sigma^+$.
Moreover, $|u|_b\geq 1$, otherwise $\mu^\infty(a) = aaaaaa\ldots$.
Since $\mu(a)$ contains both $a$ and $b$, and the morphism is not primitive, then $|\mu(b)|_a = 0$, i.e. $\mu(b) = b^\ell$ for some $\ell \geq 1$.
Finally, we can observe that if $|u|_a = 0$, then we have the morphism $\mu \equiv (ab^k, b^\ell)$ that is ultimately periodic (Proposition \ref{prop:struct_ult_per}). Therefore $|u|_a \geq 1$.
\end{proof}

\setcounter{myCounter}{\getrefnumber{prop:struct_aperiodic_non-growing}}
\begin{proofofProp}
 Let $B$ be the set of bounded letters in the morphism $\mu$.
    By Theorem \ref{th:pansiot}, we know that any morphism that generates an infinite word $x=\mu^\infty(a)$ with factor complexity $f_x(n) = \Theta(n^2)$ is non-growing and $x$ contains infinitely many runs of characters over $B^*$ with unbounded length.
    By Lemma \ref{le:non-primitive-aperiodic}, we know that $\mu \equiv (au', b^\ell)$, for some $u'$ such that $|u'|_a,|u'|_b\geq 1$ and $\ell \geq 1$.
    Since $\mu$ is growing on $a$, we have that $\ell = 1$ (that is $\mu(b) = b$).
    Note that, from Lemma \ref{prop:R-auabk_b}, we can see that, for any $i\geq1$, if $\mu \equiv (auab^k, b)$ for some $k\geq 1$, then $R_i$ grows at each iteration and therefore we have factors with unbounded length of $B^*$ in the fixed point, while on the other hand if $\mu \equiv (auba^j, b)$ then $R_i \in \mathcal{O}(1)$ and so the factors of $B^*$ in $\mu^i(a)$. 
    By \cite[Theorem 4.1]{Pansiot_ICALP84}, we know that if the number of factors of $\mu^\infty(a)$ in $B^*$ is bounded, then there exists a growing morphism $\mu':A'^* \mapsto A'^*$ and a morphism $\tau: A'^* \mapsto A^*$ such that, for some $a' \in A'$, $\tau(\mu'^\infty(a')) = \mu^\infty(a)$ and $f_{\mu^\infty(a)}(n) = \Theta(f_{\mu'^\infty(a')})$.
    Using the construction of $\mu'$ described in the proof of \cite[Theorem 4.1]{Pansiot_ICALP84} and by 
    Lemma \ref{prop:R-auabk_b}, we can deduce that $|A'|=|R_1|\geq 2$. In fact, $A'=\{[ab^ja], j\in R_1\}$, since $a$ is the only growing letter. Moreover, if $\mu^\infty(a)$ starts with $ab^{j_0}a$, with $j_0\geq 0$, then $a'=[ab^{j_0}a]$. Since $\mu$ is prolongable on $a$, the action of $\mu'$ on each letter of the alphabet $A'$ must begin with $a'$ and must contain all the letters of $A'$. It follows that $\mu'$ has to be primitive, i.e. $f_{\mu^\infty(a)}(n) = \Theta(n)$.
\qed
\end{proofofProp}

\setcounter{myCounter}{\getrefnumber{prop:struct_aperiodic_growing}}
\begin{proofofProp}
We can observe that, for any $i>0$, $$\mu_b(i)=|\mu^i(b)|=\ell^i$$ and $$\mu_a(i)=|\mu^i(a)|= |\mu^i(a)|_a + |\mu^i(a)|_b.$$ Moreover, $|\mu^i(a)|_a=n_a^i$ and $|\mu^i(a)|_b=|\mu^{i-1}(a)|_b\ell + n_b n_a^{i-1}$, where $n_a$ (resp. $n_b$) is the number of occurrences of $a$ (resp. $b$) in $\mu(a)$. 

Let us consider the case $n_a=\ell$. In this case $\mu_a(i)=in_b\ell^{i-1}+\ell^i$. Hence, $\mu_a(i)= \Theta(i\ell^i)$. Let us suppose now that $n_a\neq \ell$.

Let $F_a$ (resp. $F_b$) be the generating function of the sequence $\mu_a(i)$ (resp. $\mu_b(i)$), precisely:

$$ F_a=\sum_{i\geq 0} \mu_a(i) x^i, \, \, F_b=\sum_{i\geq 0} \mu_b(i) x^i \, . $$

Using standard methods we obtain the following system:
\begin{eqnarray}
F_a &=& 1 + n_ax F_a + n_bx F_b\\
F_b &=& 1 + \ell x F_b
\end{eqnarray}

Concerning $F_a$, we have:

$$
F_a=\frac{1}{1-n_ax} + \frac{n_bx}{(1-\ell x)(1-n_ax)} \, .
$$

This can be rewritten as

$$
F_a=\frac{1}{1-n_ax} + \frac{1- \left(  \frac{n_a (n_a-n_b-\ell)}{n_a-\ell}   \right) x}{1-n_ax}   -  \frac{1- \left(  \frac{l (n_a-n_b-\ell)}{n_a-\ell}  \right) x}{1-\ell x}  \, .
$$

We recall that the $n$th coefficient of the series $\frac{1-Qx}{1-n_ax}$ is given by the term $n_a^{n-1} (n_a-Q)$, whence

$$
\frac{1- \left(  \frac{n_a (n_a-n_b-\ell)}{n_a-\ell}   \right) x}{1-n_ax}  = \sum _{n\geq 0} \left( n_a- \frac{n_a (n_a-n_b-\ell)}{n_a-\ell} \right) n_a^{n-1} x^n = \sum _{n\geq 0} \frac{n_b}{n_a - \ell} \, n_a^{n} x^n,
$$

and similarly

$$\ \frac{1- \left(  \frac{\ell (n_a-n_b-\ell)}{n_a-\ell}  \right) x}{1-\ell x} = \sum _{n\geq 0} \frac{n_b}{n_a - \ell} \, \ell^n x^n, $$

Therefore the $i$th coefficient of $F_a$ is precisely:

$$
n_a^i+ \frac{n_b}{n_a - l} \, \left( n_a^{i} - l^i \right) \, .
$$

Hence, we can see that $i$th coefficient of $F_a$ grows as $\Theta(\ell^i)$ iff $n_a < \ell$ and $\Theta(n_a^i)$ iff $n_a > \ell$. So, we have that $\mu$ is quasi-uniform or exponentially divergent, respectively. The thesis follows by Theorem \ref{th:pansiot}.\qed
\end{proofofProp}

\setcounter{myCounter}{\getrefnumber{prop:linear_R_i}}
\begin{proofofProp}
By definition, we have that $n_a \geq 2$ and
by Lemma \ref{prop:R-auabk_b} (case {\it 2.}) and Lemma \ref{prop:R-auabk_bl}, we can see that exists at least one $1 \leq j \leq n_a -1$ such that, for any $p_h = \frac{\ell^{h-1}((\ell-1)t_j +k) - k}{\ell-1}$ with $1 \leq h \leq i$ and $p_{h}$ that grows as $h$, $ab^{p_h}a\in \mathcal{C}(\mu^i(a))$, i.e. $R_i$ contains at least $i$ elements.
Moreover, notice that $R_i$ contains at most $i(n_a - 1) + 1$ distinct elements.
It follows that $i \leq |R_i| \leq (n_a - 1)i + 1$, and therefore $|R_i| = \Theta(i)$.
\qed
\end{proofofProp}

\setcounter{myCounter}{\getrefnumber{lemma:b^(k)f(v)-general}}
\begin{proofofLemma}
If $v$ is a circular bispecial factor in $\mu^{n-1}(a)$, then we can suppose  w.l.g. $ava$ and $bvb$ are circular factors of $\mu^{n-1}(a)$.
    This means that $\mu(ava)=auab^k\mu(v)auab^k$ and $\mu(bvb) = b^\ell \mu(v)b^\ell$ are circular factors of $\mu^n(a)$. Note that any $\beta$ is circularly preceded by $b^k$, i.e. $a b^k\mu(v)a$ and $b b^k \mu(v) b$ are circular factors of $\mu^n(a)$.
\qed
\end{proofofLemma}

In order to give a proof for Lemma \ref{le:structural_general} we need the result of the following lemma.

\begin{lemma}
\label{le:factors_b^m+1_general}
Let $\mu \equiv (\alpha, \beta) = (auab^k, b^\ell)$ for some $k\geq 0$, $\ell \geq 1$ such that $k+\ell > 1$, and let $m$ be the length of the longest runs of $b$ that occurs in $auab^k$.
Then, for any $i \geq 2$, every circular factor $v$ of $\mu^i(a)$ of length $|v| > |\mu(\alpha)| + (2-\ell)m - k + \ell^2 M$ contains $b^{m+1}$ as factor, where $M = \max\{\lfloor \frac{m-(\ell+1)k}{\ell^2} \rfloor, 0\}$. 
\end{lemma}

\begin{proof}
    Let $n_a = |\mu(a)|_a$.
    Then we can uniquely factorize $\mu(a) = ab^{p_1}ab^{p_2}a \ldots ab^{p_{n_a}}$, for some $p_1, p_2, \ldots, p_{n_a - 1} \geq 0$ and $p_{n_a} = k \geq 0$.
    Let $j$ be an index such that $p_j = m$.
    Since $\mu^i(a) \in \{\mu(\alpha), \mu(\beta)\}^*$ for any $i \geq 2$, where $\mu(\alpha) = \prod_{i=1}^{n_a} auab^{p_{n_a} + \ell p_i}$ and    $\mu(\beta) = b^{\ell^2}$, we have that $b^{\ell p_j + p_{n_a}}$ occurs in $\mu(\alpha)$, for any $1 \leq j \leq n_a$.
    Note that, with the exception of the case $k=0$ and $\ell=1$ (that we are not considering since otherwise $k + \ell = 1$), $\ell p_j + p_{n_a} = \ell m + k > m$, that is $\mu(\alpha)$ has at least an occurrence of the factor $b^{m+1}$.
    Then the longest factor that may not contain $b^{m+1}$ in any string $x \in \{\mu(\alpha),\mu(\beta)\}^*$ is a factor of $\mu(\alpha)\mu(\beta)^{M} \mu(\alpha)$, where $M = \max\{\lfloor \frac{m-(\ell+1)p_{n_a}}{\ell^2} \rfloor, 0\}$  and it is 
    $v^* = b^{p_j}\prod_{i= j+1}^{n_a}auab^{p_{n_a} + \ell p_i} \cdot (b^{\ell^2})^M \cdot \prod_{i= 1}^{j-1}auab^{p_{n_a} + \ell p_i} \cdot auab^{p_j}$, since if we extend either on the left or on the right we find another $b$ and we have $b^{p_j + 1} = b^{m+1}$ as factor.
    We can see that $|v^*| = |\mu(\alpha)| + (2-\ell)p_j - k + \ell^2 M$.
\end{proof}

\setcounter{myCounter}{\getrefnumber{le:structural_general}}
\begin{proofofLemma}
 Consider the case $\ell = 1$.
 By using the proof of Lemma \ref{le:factors_b^m+1_general}, the circular factors of maximal length that do not contain $b^{m+1}$ in any string in $\{\mu(\alpha),\mu(\beta)\}^*$ must be factor of $\mu(\alpha)\mu(\beta)^{h} \mu(\alpha)$, $h\leq M = \max\{m-2k, 0\}$.
 So,  if $w$ is not a circular factor of $\bigcup_{h=0}^{M}\{\mu(\alpha)\mu(\beta)^{h}\mu(\alpha)\}$, then it holds that $w=yb^{m+1}z$, for some $y, z \in A^*$.
Note that the last $m - k + 1$ letters of $b^{m+1}$ can be obtained only as $\mu(b)$.
Since $\{\mu(a),\mu(b)\}$ is a prefix code, $z$ can be uniquely factorized as a sequence of $\mu(a)$'s and $\mu(b)$'s, i.e. $w = y b^k \mu(b)^{m-k+1} \mu(v) z'$, for some $v \in \mathcal{C}(\mu^{n-1}(a))$ and some proper prefix $z'$ of $\mu(a)$ or $\mu(b)$.
From the fact that $w$ is right special, we can prove that $z'=\varepsilon$. 
In fact, if $z' \neq \varepsilon$, then $z'$ is a proper prefix of $\alpha$ and $w$ must be followed by $\alpha_{|z'| + 1}$, i.e. $w$ would not be right special.
Moreover, $y b^k \mu(b^{m-k+1})$ can be uniquely factorized as a sequence of $\mu(a)$'s and $\mu(b)$'s, unless for a prefix $b^j$, with $j\leq k$, since that prefix could be a proper suffix of $\alpha$, a run of $\beta$'s or a combination of both. 
However, if $j <k$, then $w$ must be preceded by $b^{k-j}$, that is $w$ is not left special and therefore $j=k$.

If $\ell > 1$, the proof holds (recall that $M= \max\{\lfloor\frac{m-(\ell-1)k)}{\ell^2}\}$ ), but with the exception of any bispecial factor $w = b^h$, for any $h \geq m+1$, that do not occur in $\bigcup_{h=0}^M\{\mu(\alpha)\mu(\beta)^h\mu(\alpha)\}$.
In fact, only in this case we can not uniquely factorize any of the factor of $w$ (observe that $b\mu(\beta) = \mu(\beta)b$).
 Let $m'$ be the greatest value such that $ab^{m'}a$ is a factor of $\mu^n(a)$.
 Since $ab^{m'}a = ab^h b^{m'-h}a = ab^{m'-h}b^h a$ for any $1\leq h < m'$, it holds that $ab^h b$ and $bb^h a$ occur in $\mu^n(a)$.
 It follows that $b^h$ is bispecial for any $0 \leq h < m'$ and the thesis follows.
\qed
\end{proofofLemma}

\setcounter{myCounter}{\getrefnumber{prop:upper_bound}}
\begin{proofofTh}
Let $m$ be the longest run of $b$'s that occurs in $auab^k$ and let
$M=\max\{\lfloor \frac{m-(\ell+1)k}{\ell^2} \rfloor, 0\}$.
Let $BS_0(\mu^i(a)) = \{v \in \{a,b\}^* \mid v$ is a bispecial circular factor of $\bigcup_{i=0}^{M}\{\mu(\alpha)\mu(\beta)^{i}\mu(\alpha)\}\}$, $BS_b(\mu^i(a)) = \{b^h \mid b^h \in \mathcal{C}(\mu^i(a))\}$ and $BS_\mu(\mu^i(a)) = \{b^k\mu(v) \mid v \in \mathcal{C}(\mu^{i-1}(a))$ and $|v|_a\geq1\}$.

From Lemma \ref{le:structural_general}, we can see that $BS(\mu^i(a)) = BS_0(\mu^i(a)) \cup BS_b(\mu^i(a)) \cup BS_\mu(\mu^i(a))$ (note that the intersection can be non-empty).
It is easy to see that $|BS_0| = \mathcal{O}(1)$.
Moreover, since $\mu^i(a) = au\mu(u)\mu^2(u)\ldots\mu^{i-1}(u)$, we can see that any element $w \in BS_\mu(\mu^i(a))$ that do not belong to $BS_\mu(\mu^{i-1}(a))$ has to cross $\mu^{i-1}(u)$ and $w=b^k\mu(v)$ for some $v$ that crosses $\mu^{i-2}(u)$. 
Iterating the procedure, we can see that $|BS_\mu(\mu^i(a))| = \mathcal{O}(i)$.
If $\ell=1$, from Lemma \ref{le:structural_general}, we can see that also $BS_b(\mu^i(a))= \mathcal{O}(i)$.
It follows that $r_\bwt(\mu^i(a)) \leq |BS(\mu^i(a)| \leq |BS_0(\mu^i(a)| + |BS_\mu(\mu^i(a)| + |BS_b(\mu^i(a)| = \mathcal{O}(i)$.

On the other hand, if $\ell >1$, clearly
$|BS_b(\mu^i(a))| = \Omega(\ell^i)$, and therefore $|BS(\mu^i(a))| = \Omega(\ell^i)$.
However, recall that each change of letters in $\bwt(w)$ uniquely corresponds to a unique bispecial factor, that is the longest common prefix between the two conjugates in correspondence of the change of letter (see Lemma \ref{le:r_bound_bispecial} and recall that $e_l(u),e_r(u) \leq 1$ for binary alphabets).
Let $m'$ be the length of the longest run of $b$'s in $\mu^i(a)$.
Note that $m' \in R_i$.
We define $\oR_i = \{0, 1, \ldots, m'-1\} \setminus R_i$, i.e. $\oR_i$ is the set of indices $0 \leq \oq < m'$ such that $ab^{\oq}a$ is not a factor of $\mu^i(a)$.
By Lemma \ref{prop:R-auabk_bl} we know that $|R_i|$ grows as $\mathcal{O}(i)$.
This implies that, even though $|\oR_i| \in \Omega(\ell^i)$, we can split $\oR_i$ in $K \leq |R_i|$ subsets of maximal sizes that contain consecutive lengths of runs of $b$'s, i.e. $P_j =\{h_j, h_j + 1, \ldots, h_j + d_j \in \oR_i \mid h_j - 1,h_j + d_j + 1\notin \oR_i\}$, for any $1 \leq j \leq K$.
We can see that any $b^h$ is the longest common prefix of two consecutive conjugate words $C_t$ and $C_{t+1}$ of $w$, for some $0\leq t < n$ if and only if $C_t = b^h a z$ and $C_{t+1} = b^{h+1} a z'$, for some $z, z' \in \Sigma^*$ (if $b<a$ then just switch the order of $C_t$ with $C_{t+1}$).
Note that we can make correspond to each set $P_j$ the maximal interval $[s_j\ldots s_j + \ell_j]$ (for some $0 \leq s_j, \ell_j \leq |\mu^i(a)|$) of sorted conjugates $\{C_{s_j}, C_{s_j + 1}, \ldots, C_{s_j + \ell_j}\}$ with prefix with $b^{h}a$, for any $h \in P_j$.
Let $y = \bwt(\mu^i(a))$.
Since by definition $ab^h a$ is not a factor of $\mu^i(a)$ for any $h \in \oR_i$, then $y_{s_j}y_{s_j+1}\cdots y_{s_j + \ell_j} = b^{\ell_j + 1}$, and therefore a change of letter in $y$ can not correspond to a bispecial factor $b^h$ with $h \in P_j$, with the only possible exception for $b^{h_j + d_j}$ if $y_{s_j + \ell_j +1} = a$, since $ab^{h_j + d_j + 1}a$ is a factor of $\mu^i(a)$ and $lcp(C_{s_j + \ell_j},C_{s_j + \ell_j + 1}) = b^{h_j + d_j}$.
Finally, since there are $K \leq |R_i|$ of these intervals, we have that $r_\bwt(\mu^i(a)) 
\leq |BS_0(\mu^i(a))| + |BS_\mu(\mu^i(a))| + 2|R_i| = \mathcal{O}(i)$.
\qed
\end{proofofTh}

\setcounter{myCounter}{\getrefnumber{prop:lowerbound_r_Rn}}
\begin{proofofTh}
Let $C_1 < C_2 < \ldots < C_n$ be the lexicographically sorted rotations of $\mu^i(a)$ and let $m'$ be the length of the longest run of $b$'s in $\mu^i(a)$.
We define the set of $(s_j, \ell_j)$-pairs  that define the intervals $X_j = \{C_{s_j}, C_{s_j + 1}, \ldots, C_{s_j + \ell_j - 1}\}$, where $X_j$ contains all and only the rotations with prefix $b^{j}a$, for any $0 \leq j \leq m'$.
Let $y = \bwt(\mu^i(a))$.

We can see that for any $0 < j \leq m'$, if $X_j$ exists, then $y_{s_{j-1}}\cdots y_{s_{j-1}+\ell_{j-1}}$ contains at least $\ell_j$ occurrences of $b$'s.
In fact, for any $h \geq 0$, if exist $\ell_j$ rotations such that $C_j = b^{h+1}a u$ for some $u \in A^*$, then exist as much rotations  $C_{j'} = b^haub$.
It follows that, for any $0 \leq j < m'$, $y_{s_{j}}\cdots y_{s_{j}+\ell_{j}-1}$ contains at least one $b$.

Moreover, by definition, we know that if $j \in R_i$, then there is at least an $a$ in $y_{s_j}\cdots y_{s_j+\ell_j-1}$.
Hence, for any $0 \leq j' \leq \lceil \frac{|R_i|-1}{2}\rceil$ such that $q_{2j'+1} \in R_i$, we have that in $y_{s_{2j'+1}} \cdots  y_{s_{2(j'+1)+1}+ \ell_{2(j'+1)+1}-1}$ it must occur $aba$ as subsequence, since there is at least one run of $a$'s in $y_{s_{2j'+1}}\cdots  y_{s_{2j'+1}+ \ell_{2j'+1}}$, at least one run of $b$'s in $y_{s_{2(j'+1)}} \cdots  y_{s_{2(j'+1)}+ \ell_{2(j'+1)}-1}$ and at least one run of $a$'s in $y_{s_{2(j'+1)+1}}\cdots  y_{s_{2(j'+1)+1}+ \ell_{2(j'+1)+1}-1}$.
Hence, $r_\bwt(\mu^i(a)) \geq r(y_{s_1}\cdots y_{s_{m'}+\ell_{m'}-1}) \geq r((ab)^{\lceil\frac{|R_i|-1}{2}\rceil}a) = \Omega(|R_i|)$.
Finally, by Proposition \ref{prop:struct_aperiodic_non-growing} we know that $f_x(n) = \Theta(n^2)$, if and only if $\mu \equiv (auab^k,b)$ for any $u \in A^*$ and $k \geq 0$, while if $\mu$ is growing and non-primitive then, by Proposition \ref{prop:struct_aperiodic_growing}, $\mu \equiv (av,b^\ell)$ for some $\ell \geq 2$
\qed
\end{proofofTh}

\setcounter{myCounter}{\getrefnumber{lb}}
\begin{proofofLemma}
We observe that $\mu^i(a)$ has as many $\alpha$'s and $\beta$'s as $a$'s and $b$'s in $\mu^{i-1}(a)$, that is $\alpha_i$ and $\beta_i$ respectively, and since $\mu^i(a)\in \{\alpha, \beta\}^{|\mu^{i-1}(a)|}$  we can have at most $|\mu^{i-1}(a)|-1$ images of a letter that merge their last equal-letter run with the next one.
As regards $\alpha_i$ and $\beta_i$, by definition of $\mu$ we get:
$$\begin{cases}
\alpha_1=1\\
\beta_1=0\\
\alpha_i = n_a \alpha_{i-1} + m_a \beta_{i-1}\\
\beta_i = n_b \alpha_{i-1} + m_b \beta_{i-1}
\end{cases}$$
By recursion we have:
\begin{align*}
    \alpha_i &= n_a\alpha_{i-1}+m_a\beta_{i-1}\\
            & = n_a\alpha_{i-1}+m_a(n_b \alpha_{i-2} + m_b \beta_{i-2})=n_a\alpha_{i-1}+m_an_b \alpha_{i-2} + m_am_b \beta_{i-2}\\
             &=n_a\alpha_{i-1}+m_an_b \alpha_{i-2} +m_am_bn_b\alpha_{i-3}+m_am_bm_b\beta_{i-3}\\
             & = \cdots \\
             & = n_a \alpha_{i-1} + \sum_{j=0}^{i-3}m_a n_b m_b^{j} \alpha_{i-2-j}
\end{align*}
and
\begin{align*}
    \beta_i &=m_b\beta_{i-1} + n_b \alpha_{i-1}\\
            & =m_b\beta_{i-1}+n_b(m_a \beta_{i-2} + n_a \alpha_{i-2})=m_b\beta_{i-1} + m_an_b \beta_{i-2}+n_an_b \alpha_{i-2}\\  
             &=m_b\beta_{i-1}+m_an_b \beta_{i-2} +m_an_an_b\beta_{i-3}+n_a^2n_b\alpha_{i-3}\\
             &= \cdots\\
             & = m_b \beta_{i-1} + \sum_{j=0}^{i-4}m_a n_a^{j} n_b  \beta_{i-2-j} + n_a^{i-2}n_b\\
\end{align*}
From the above considerations, the thesis follows.
\qed
\end{proofofLemma}

\setcounter{myCounter}{\getrefnumber{th:highlyBWTcompres}}
\begin{proofofTh}
Let us suppose that $\mu$ is prolongable on $a$- The proof is analogous if $\mu$ is prolongable on $b$. If $\mu^\infty(a)$ is ultimately periodic, then $r_{\bwt}(\mu^i(a))$ is $\Theta(1)$. By Propositions \ref{prop:struct_ult_per} and \ref{prop:exponential_growth_run} we have that $r(\mu^i(a))$ is exponential, except the case $\mu\equiv(ab^m,b^n)$. Then $\lim_{i\to\infty}\rho(\mu^i(a))=0$. If $\mu\equiv(ab^m,b^n)$, then $\rho(\mu^i(a))$ is $\Theta(1)$. Let us suppose $\mu^{\infty}(a)$ is aperiodic. If $\mu$ is primitive, from Corollary \ref{cor:primitive_morph} the thesis follows. If $\mu$ is not primitive, then $r_{\bwt}(\mu^i(a))$ is $\mathcal{O}(i)$ (Theorem \ref{prop:upper_bound}). Moreover, by using Lemma \ref{lb}, Propositions \ref{prop:struct_aperiodic_non-growing} and \ref{prop:struct_aperiodic_growing}, we have that $r(\mu^i(a))$ is $\Omega(2^i)$, therefore the thesis follows.\qed
\end{proofofTh}

\end{document}